\documentclass[conference]{IEEEtran}

\usepackage{amsfonts, amssymb, amsmath, amsthm, mathtools, paralist}
\usepackage{xcolor, pagecolor, lipsum} 
\usepackage{cleveref}
\usepackage{cite}
\usepackage{booktabs}

\usepackage[colorinlistoftodos,backgroundcolor=lightgray!50,linecolor=black,textsize=scriptsize]{todonotes}

\DeclarePairedDelimiter\parens{\lparen}{\rparen}
\DeclarePairedDelimiter\angles{\langle}{\rangle}
\DeclarePairedDelimiter\bracks{\lbrack}{\rbrack}
\newcommand{\R}{\mathbb{R}}
\newcommand{\Ker}{\operatorname{Ker}}
\renewcommand{\Im}{\operatorname{Im}}

\newcommand{\calN}{\mathcal{N}}
\newcommand{\half}{\frac{1}{2}}
\newcommand{\UI}{\operatorname{UI}}
\newcommand{\SI}{\operatorname{SI}}

\newcommand{\TMXY}{\text{TMXY}}
\newcommand{\MYXT}{\text{MYXT}}
\newcommand{\indep}{\perp\!\!\!\!\perp} 

\newcommand{\GacsKorner}{Gács-Körner}

\newtheorem{theorem}{Theorem}[subsection]
\newtheorem{lemma}{Lemma}[subsection]
\newtheorem{proposition}{Proposition}[subsection]
\newtheorem{conjecture}{Conjecture}[subsection]
\theoremstyle{definition}
\newtheorem{definition}{Definition}[subsection]
\newtheorem{remark}{Remark}[subsection]

\title{Extracting Unique Information\\Through Markov Relations}

\author{\IEEEauthorblockN{Keerthana Gurushankar\textsuperscript{1}, Praveen Venkatesh\textsuperscript{3} and Pulkit Grover\textsuperscript{1,2}}
\IEEEauthorblockA{
\textsuperscript{1}Dept.~of Electrical and Computer Engineering, 
\textsuperscript{2}Neuroscience Institute, Carnegie Mellon University, Pittsburgh, PA \\
\textsuperscript{3}Allen Institute \& University of Washington, Seattle, WA}
}

\begin{document}

\maketitle
\begin{abstract}
    We propose two new measures for extracting the unique information in $X$ and not $Y$ about a message $M$, when $X, Y$ and $M$ are joint random variables with a given joint distribution. 
    We take a Markov based approach, motivated by questions in fair machine learning, and inspired by similar Markov-based optimization problems that have been used in the Information Bottleneck and Common Information frameworks.
    We obtain a complete characterization of our definitions in the Gaussian case (namely, when $X, Y$ and $M$ are jointly Gaussian), under the assumption of Gaussian optimality.
    We also examine the consistency of our definitions with the partial information decomposition (PID) framework, and show that these Markov based definitions achieve non-negativity, but not symmetry, within the PID framework.
\end{abstract}

\section{Introduction}

Consider the problem of designing machine learning algorithms that can make fair decisions despite biases present in training data.
Concretely, suppose we are given a feature set $X$, from which we need to extract a decision $T$ that captures, as accurately as possible, the true labels of the data $M$, while avoiding dependence on certain \emph{protected attributes} $Y$, such as gender or race.
Producing a fair decision, in this instance, is equivalent to extracting \emph{that} information from the data $X$ about the true labels $M$, which is \emph{independent} of $Y$.
This can be seen as extracting information from $X$ about $M$ that is \emph{uniquely} present in $X$ and not in $Y$. 
Thus, we want to compute a random variable $T$ from $X$ through a Markov Chain $T$---$X$---$(M,Y)$, where $T$ captures the maximum possible information relevant to $M$, $I(T; M)$, while being independent of $Y$, i.e., $T\indep Y$.
An important question that arises in this context is, what is this variable $T$, and how much information about $M$ can it contain?

An alternative, but closely related problem, arises in the context of ``post-processing'' the decision of a machine learning algorithm, e.g., when the algorithm is a black box to the fairness engineer.
Here, given the decision of an algorithm, $M$, the goal is to process it into some variable $T$ to make the dependence on $Y$ (the protected attribute) small, while keeping it as faithful to $X$ (the features) as possible.
One motivation for keeping $T$ faithful to $X$ could be that the original algorithm was trained to generate an $M$ that extracted the information from $X$ that is most relevant to an accurate decision, and now we want to keep that information while excluding information from $Y$.
Thus, we wish to extract those parts $T$ of the true labels $M$ through the Markov chain $T$---$M$---$(X, Y)$, which contain as much information as possible about $X$, while being independent or agnostic to $Y$.
So once again, $T$ has only information about $M$ that is \emph{uniquely} present in $X$ but not in $Y$.

The question of quantifying how much information about $M$ is uniquely present in $X$ and not in $Y$ has been addressed in the recent and growing literature on Partial Information Decomposition (PID).
The PID framework seeks measures that explain the information interactions of two random variables $X$ and $Y$ regarding a third variable $M$. More precisely, it decomposes the mutual information \(I(X,Y; M)\) into four parts: (i) redundant information (\(R\)) about $M$ that can be extracted from either \(X\) or \(Y\), (ii) unique information about $M$ in $X$ ($\UI_X$), which is not in $Y$, (iii) unique information about $M$ in $Y$ ($\UI_Y$) not available in $X$, and (iv) synergistic information $S$ about $M$, which can only be obtained from the combined pair $X$ and $Y$.
First proposed in the work of Williams and Beer~\cite{williams2010nonnegative}, this gives us the fundamental ``sum''-axioms for PID:
\begin{align}
    I(M; X,Y) &= R + \UI_X + \UI_Y + S, \label{eq:pid1} \\
    I(M; X) &= R + \UI_X, \label{eq:pid2} \\
    I(M; Y) &= R + \UI_Y. \label{eq:pid3}
\end{align}
Here,~\eqref{eq:pid2} and~\eqref{eq:pid3} are seen as natural extensions of the intuitive explanations for these quantities: the mutual information that $X$ (alone) has about $M$ should equal the amount that can be uniquely extracted from $X$ plus that which is redundantly present in both $X$ and $Y$.
In addition, one also requires that the four partial information components, $R$, $\UI_X$, $\UI_Y$ and $S$, be non-negative, so that they may be interpreted meaningfully as information quantities.
% \begin{align}
%     I(M; X|Y) &= \UI_X + S\\
%     I(M; Y|X) &= \UI_Y + S
% \end{align}
% \todo{explain them?}

These sum and non-negativity axioms form the basis of the PID framework; however, they do not uniquely specify a definition for the four partial information components.
Therefore, various definitions of these partial information components have been proposed, with different goals and approaches~\cite{williams2010nonnegative, harder2013bivariate, bertschinger2014quantifying, banerjee2018unique, niu2019measure} (see \cite{lizier2018information} for a review).

While these definitions were obtained using an axiomatic development and to maintain consistency with certain canonical examples, they fall short of our requirements in one important dimension: they do not compute a specific random variable $T$.
Instead, they quantify the \emph{amount} of mutual information that is unique, redundant or synergistic. \footnote{Thus, these definitions could be said to be more \emph{statistical} than \emph{structural}. Harder et al.~\cite{harder2013bivariate} discuss a similar idea very briefly in the context of measuring redundancy, calling it ``mechanistic'' vs.\ ``source'' redundancy rather than statistical vs.\ structural.}

In this paper, we make the following contributions:
\begin{enumerate}
    \item We propose two definitions that are motivated by each of the two fairness problems described at the start of this section. These definitions provide two ways to quantify unique, redundant, and synergistic information while also providing a random variable $T$ (jointly distributed with $M$, $X$ and $Y$) that \emph{represents} the unique information.
    %\item We discuss when each of these two definitions is more suitable, by considering several examples.
    \item We also analyze the case where $M$, $X$ and $Y$ are jointly Gaussian, and obtain closed-form expressions for the unique information in terms of the kernel and image spaces of the covariance matrices involved.
    These results are derived under the assumption that the optimal solution is also Gaussian (formally specified later).
    \item Examining consistency with the PID framework, we find our definition also satisfies the non-negativity axiom, although it does not give rise to a symmetric redundancy (we show this using a counterexample in the Gaussian case).
\end{enumerate}
Our definitions are motivated by Markov-based problem formulations proposed earlier in the literature, such as in the information bottleneck framework~\cite{tishby2000information} and in the definitions of Wyner~\cite{wyner1975common} and \GacsKorner{}~\cite{gacs1973common} common information (the latter two seek meaningful decompositions of \emph{two} random variables into common and unique parts).
In essence, a Markov relationship is assumed between the solution variable (here, $T$) and the variables that specify the problem ($M$, $X$ and $Y$): this means $T$ can be seen as a stochastic function of one or more of $M$, $X$ or $Y$.
Inspired by these frameworks as well as PID-like quantities in applications, we propose a natural Markov based approach to PID.

The two problems we discuss are based on the Markov chains $T$---$X$---$(M, Y)$ and $T$---$M$---$(X, Y)$, and we define only the \emph{unique information} in both settings. Both these problems require us to extract the unique information present in $X$ but not $Y$ about $M$, and produce theoretically equivalent optimization problems, differing only in the interchange of message $M$ and variable $X$.
Although these optimization problems are equivalent for the purpose of analysis, in practice, an interchange of the variables $M$ and $X$ may be highly non-trivial, for e.g. when the message and source variables are of very different dimensionalities.
They also have different properties and meanings, when considering the full PID framework (redundancy and synergy, rather than merely unique information).

One important limitation of our definitions is that the unique information for $X$ and for $Y$ ($\UI_X$ and $\UI_Y$) do not lead to a symmetric notion of redundancy (or synergy).
That is, if we define redundancy as the difference of the mutual information and the unique information (refer Eqs.~\eqref{eq:pid2}, \eqref{eq:pid3}), we find that
\begin{equation*}
    R_X \coloneqq I(M ; X) - \UI_X \neq I(M; Y) - \UI_Y \eqqcolon R_Y,
\end{equation*}
in general.
However, this is a common problem across many PID definitions (e.g., \cite{harder2013bivariate, banerjee2018unique, venkatesh2022partial}), and is commonly dealt with by assigning $R$ to be the minimum of $R_X$ and $R_Y$.
We do not explicitly define the PID this way, however, since it disrupts the interpretation of the variable $T$, which encodes the part of the information about $M$ that is uniquely present in $X$ and not in $Y$.
We leave a more detailed analysis of the possibility of symmetrization and its effects to future work.

The remainder of the paper is organized as follows: in Section II, we propose the two Markov-based measures for extracting unique information.
We also define variants for the Gaussian case, restricting the search space of Markov chains to jointly Gaussian random variables.
In Section III, we fully solve this Gaussian subvariant of the unique information problem to obtain a closed form solution in terms of the covariance matrices of the joint distribution. 
Further, looking beyond unique information, in Section IV, we examine the consistency of these Markov-based definitions with the PID framework, namely their compatibility with the sum, non-negativity and symmetry axioms proposed in Williams and Beer~\cite{williams2010nonnegative} as well as the matching of intuitive expectations in simple logical operator examples.

\section{Notation and Definitions}\label{sec:definitions}
\subsection{Notation}
Upper case letters (e.g. $T,M,X,Y$) are used to denote random variable/vectors. Random vectors are assumed to be column vectorized. $\indep$ is used to denote independence and conditional independence relationships. The notation $X-Y-Z$ denotes a Markov chain between the three random variables (i.e., $X\indep Z|Y$). $\Sigma$ is used to denote covariance and cross-covariance matrices. E.g., for random vectors $M$ and $X$, $\Sigma_{M}$ denotes the covariance $\mathbb{E}\left[MM^T\right]$,  $\Sigma_{MX}$ denotes the cross-covariance $\mathbb{E}\left[MX^T\right]$, and  $\Sigma_{(MX)}$ denotes the covariance of the appended vector $[M^T;X^T]^T$, i.e., $\mathbb{E}\left[[M^T;X^T]^T[M^T;X^T]\right]$ (for zero mean random vectors). $\UI_X = \UI(M:X\backslash Y)$ and $\UI_Y = \UI(M:Y\backslash X)$ are both used interchangeably to refer to unique information, as per the level of detail needed. 

\subsection{Definitions of Unique Information}

\begin{definition}[TMXY-Unique Information] \label{def:tmxy} If message \(M\), and targets \(X, Y\) are random variables jointly distributed as \(p(m, x, y)\), the \(TMXY\)-Unique Information present in \(X\) but not \(Y\) about \(M\) is defined as 
	\[
		\UI(M: X\backslash Y) = \max_{T-M-XY, \ T \indep Y} I(T; X).
	\]
\end{definition}

\begin{definition}[TMXY-Gaussian Unique Information] \label{def:tmxyG} If message \(M\), and targets \(X, Y\) are jointly Gaussian random variables, the \(TMXY\)-Gaussian Unique Information present in \(X\) but not \(Y\) about \(M\) is defined as 
	\[
		\UI^{(\calN)}(M: X\backslash Y) = \max_{\substack{T- M - XY,\  T \indep Y 
		\\ TMXY \text{ jointly Gaussian}}} I(T; X). 
	\]
\end{definition}
%\subsection{Source Markov Unique Information}
\begin{definition}[MYXT-Unique Information] \label{def:myxt} If message \(M\), and sources \(X, Y\) are random variables jointly distributed as \(p(m, x, y)\), the \(MYXT\)-Unique Information present in \(X\) but not \(Y\) about \(M\) is defined as 
	\[
		\hat{\UI}(M: X\backslash Y) = \max_{MY-X-T, \ T \indep Y} I(T; M).
	\]
\end{definition}

\begin{definition}[MYXT-Gaussian Unique Information] \label{def:myxtG}If message \(M\), and sources \(X, Y\) are jointly Gaussian random variables, the \(MYXT\)-Gaussian Unique Information present in \(X\) but not \(Y\) about \(M\) is defined as 
	\[
		\hat\UI^{(\calN)}(M: X\backslash Y) = \max_{\substack{ MY - X - T, \ T \indep Y 
		\\ MXYT \text{ jointly Gaussian}}} I(T; M). 
	\]
\end{definition}

\begin{remark}
Note that Definitions~~\ref{def:tmxyG} and~\ref{def:myxtG} are not merely special cases of Definitions~\ref{def:tmxy} and~\ref{def:myxt}. This is because the optimizations in  Definitions~~\ref{def:tmxyG} and~\ref{def:myxtG} are performed not over all joint distributions satisfying the Markov relationship, but only over distributions that are jointly Gaussian. Whether the optimal solutions for Definitions~\ref{def:tmxy} and~\ref{def:myxt} for Gaussian $M,X,Y$ turn out to be Gaussian remains to be established.
\end{remark}

\section{Results}
\begin{theorem}
	If \(M, X, Y\) are jointly Gaussian random variables, prewhitened such that \(\Sigma_{M}, \Sigma_{X}, \Sigma_{Y} \) are identity matrices of their respective dimensions, then we have 
	\begin{align*}
		\UI^{(\calN)}(&M: X\backslash Y)  \\&=-\half\log\det(I - V_{Y_\perp M}\Sigma_{MX}\Sigma_{XM} V_{M Y_\perp}).
	\end{align*}
	where \(V_{MY_\perp}\) is the orthonormal matrix whose columns span \(\Ker(\Sigma_{YM})\).
\end{theorem}
\begin{proof}
	The constraints and relevant terms in the problem assuming all variables have been pre-whitened as \(T \mapsto \Sigma_T ^{-\half}T, M \mapsto \Sigma^{-\half} M\) and similarly for \(X\) and \(Y\),  can be written in terms of fixed covariance matrices, and free parameter \(\Sigma_{MT}\) as follows:
	
	First, to check the positive semi-definiteness constraint $\Sigma_{(TMXY)} \succcurlyeq 0$, given $\Sigma_{(MXY)} \succcurlyeq 0$, it suffices to check that the Schur complement $\Sigma_{(TMXY)}/\Sigma_{(MXY)} \succcurlyeq 0$. Substituting the Markov constraint $T-M-XY$ in this matrix (namely, $\Sigma_{T(XY)} = \Sigma_{TM}\Sigma_{M(XY)}$ using \Cref{gaussian-markov}), we find this reduces to $\Sigma_{(TMXY)}/\Sigma_{(MXY)} = \Sigma_{(TM)}/\Sigma_M = I - \Sigma_{TM}\Sigma_{MT} \succcurlyeq 0$. For the remaining constraint $T\indep Y$ and the objective function, we find
	\begin{align*}
		 T \indep Y &\iff \Sigma_{YT} = \Sigma_{YM}\Sigma_{MT} = 0\\
		 I(T; X) &= -\half\log\det(I - \Sigma_{TX}\Sigma_{XT}) \\
		 &=  -\half\log\det(I - \Sigma_{TM} \Sigma_{MX}\Sigma_{XM}\Sigma_{MT})
	\end{align*}
	Thus, in terms of parameter \(\Sigma_{MT}\), \( \UI^{(\calN)}(M: X\backslash Y) \) is
	\begin{align*}
	    &\max -\half\log\det(I - \Sigma_{TM} \Sigma_{MX}\Sigma_{XM}\Sigma_{MT})\\
		\text{over } &\Sigma_{MT} \text{ such that }\Sigma_{YM}\Sigma_{MT} = 0,\ \Sigma_{TM}\Sigma_{MT} \preccurlyeq I
	\end{align*}
	Now, since \(\Sigma_{YM}\Sigma_{MT} = 0\), we must have \(\Im(\Sigma_{MT}) \subseteq \Ker(\Sigma_{YM})\). 
	We use this to parametrize \(\Sigma_{MT}\) as follows: choose an orthonormal basis for \(\Ker(\Sigma_{YM})\) and let \(V_{MY_\perp}\) be the matrix whose columns are this basis. Then we have \(\Sigma_{YM}V_{MY_\perp} = 0\) and \(V_{MY_\perp}^T V_{MY_\perp} = I_p\). Now, \(\Sigma_{YM}\Sigma_{MT} = 0\) iff \(\Sigma_{MT} = V_{MY_\perp} \Sigma\) for some matrix \(\Sigma\in \R^{p\times t}\), and this yields a parametrization \(\Sigma_{MT}(\Sigma)\). 
	\begin{align*}
		\Sigma_{TM}\Sigma_{MT} &= \Sigma^T  V_{Y_\perp M} V_{M Y_\perp } \Sigma = \Sigma^T \Sigma\\
		\Sigma_{TX}\Sigma_{XT} &= \Sigma^T V_{Y_\perp M}\Sigma_{MX} \Sigma_{XM} V_{MY_\perp} \Sigma
	\end{align*}
	Then the optimization problem, in terms of \(\Sigma\) becomes
	\begin{align*}
		&\UI^{(\calN)}(M: X\backslash Y) = \\
		&\max_{\Sigma^T \Sigma \preccurlyeq I} -\half\log\det(I - \Sigma^T  V_{Y_\perp M}\Sigma_{MX} \Sigma_{XM} V_{MY_\perp} \Sigma)
	\end{align*}
	Let $A = V_{Y_\perp M}\Sigma_{MX}\Sigma_{XM} V_{MY_\perp}$, for every $\Sigma$ such that $\Sigma^T\Sigma \preccurlyeq I$, we must have
	\[
	    \det(I-\Sigma A \Sigma^T) \geq \det(I-A)
	\]
	by reasoning using simple determinant identities as follows:
	\begin{align*}
		\det(I - \Sigma^T A \Sigma) &= \det(I- A\Sigma\Sigma^T)\\
		&= \det(I-A + A(I- \Sigma\Sigma^T))\\
		&\geq \det(I-A) + \det(A)\det(I-\Sigma\Sigma^T)\\
		&\geq \det(I-A)
	\end{align*}
	Note that $\Sigma\Sigma^T \preccurlyeq I$ precisely whenever $\Sigma^T \Sigma \preccurlyeq I$, as can be seen by reasoning with Schur complements. Note also that equality is achieved by \(\Sigma:\Sigma \Sigma^T = I\), i.e., by  \(\Sigma_{MT} = V_{MY_\perp}\). Thus the unique information is obtained to be
	\begin{align*}
		&\max_{\substack{T- M - XY,\  T \indep Y \\ TMXY \text{ jointly Gaussian}}} I(T; X) \\
		&= -\half\log\det (I -  V_{Y_\perp M}\Sigma_{MX} \Sigma_{XM} V_{MY_\perp}). 
	\end{align*}
\end{proof}

\begin{theorem}
	If \(M, X, Y\) are jointly Gaussian random variables, prewhitened such that \(\Sigma_{M}, \Sigma_{X}, \Sigma_{Y} \) are identity matrices of their respective dimensions, then we have 
	\begin{align*}
		\hat\UI^{(\calN)}(&M: X\backslash Y)\\
		&= -\half\log\det(I - V_{Y_\perp X}\Sigma_{XM}\Sigma_{MX} V_{X Y_\perp})
	\end{align*}
	where \(V_{XY_\perp}\) is the orthonormal matrix whose columns span \(\Ker(\Sigma_{YX})\).
\end{theorem}
\begin{proof}
    The proof follows very similarly to the previous result, as the optimization problems for the two definitions differ only in an interchange of the message random variable $M$ and source random variable $X$. In effect, here we obtain an optimization problem in the free parameter $\Sigma_{XT}$ that would yield $\hat\UI^{(\calN)}(M:X\backslash Y)$ as
    \begin{align*}
	    &\max -\half\log\det(I - \Sigma_{TX} \Sigma_{XM}\Sigma_{MX}\Sigma_{X})\\
		\text{over } &\Sigma_{XT} \text{ such that }\Sigma_{YX}\Sigma_{XT} = 0,\ \Sigma_{TX}\Sigma_{XT} \preccurlyeq I
	\end{align*}
	which is solved by $\Sigma_{XT} = V_{XY_\perp}$.
\end{proof}

\section{Properties}
In this section, we turn to examine the consistency of our definitions with the PID axioms. 

Beginning with our definition of unique information, we define unsymmetrized redundant and synergistic information terms, 
\begin{align*}
    R_X &= I(M; X) - \UI_X, &S_X &= I(M; X | Y) - \UI_X\\
    R_Y &= I(M; Y) - \UI_Y, &S_Y &= I(M; Y|X) - \UI_Y
\end{align*}

\subsection{Non-Negativity}

We study the consistency of these unsymmetrized terms with Williams and Beer's non-negativity axiom, namely whether $\UI_X, \UI_Y, R_X, R_Y, S_X, S_Y$ are all non-negative. 

It is easy to see that the unique information terms, being mutual information quantities, are always non-negative. Below, we show that  both redundancy and synergy terms are also non-negative in every case. 
\begin{proposition}
	\begin{align}
		\UI(M:X\backslash Y) &\leq I(M; X)
	\end{align}
\end{proposition}
\begin{proof}
    First, for the $\TMXY-\UI$, it follows from the data processing inequality that for every \(T-M-X\), we must have \(I(T; X) \leq I(M; X)\), thus taking maximum over such a set of \(T\), we must have \(\UI(M:X\backslash Y) =\max_{T-M-XY,\ T\indep Y} I(T; X) \leq I(M; X)\). 
	
	The proof for the second definition follows likewise from the data processing inequality. 
\end{proof}

\begin{proposition}
	\begin{align}
		\UI(M:X\backslash Y) &\leq I(M;X|Y)
	\end{align}
\end{proposition}
\begin{proof}
    First, for the $\TMXY-\UI$ definition, we show that for every \(T\) such that \(T-M-XY\) and \(T\indep Y\), we must have \(I(T; X) \leq I(M; X | Y)\). For this first, by the chain rule of mutual information,
	\begin{align*}
		I(T; XY) &= I(T; Y) + I(T; X | Y) = I(T; X | Y)\\
		&= I(T; X) + I(T; Y| X)
	\end{align*}
	Thus, \(I(T; X) \leq I(T; X | Y)\), so it suffices to show \(I(T; X | Y) \leq I(M; X | Y)\). For this, we use
	\begin{align*}
		I(X; TM | Y) &= I(T; X | Y) + I(X; M | TY)\\
		&= I(X; M | Y) + I(X; T | MY)
	\end{align*}
	So it would suffice to show \(I(X; T|MY) = 0\) to obtain the desired inequality. To this end, we show that \(I(X; TMY) = I(X; MY)\) as follows:
	\begin{align*}
		I(X; TMY) &= I(X; M) + I(X; YT | M)\\
		&= I(X; M) + I(X; T | M ) + I(X; Y | MT)\\
		&= I(X; M) + I(X; Y | M)\\
		&= I(X; YM)
	\end{align*}
	And thus, \(I(X; T | MY) = I(X; TMY) - I(X; MY) = 0\). 
	
	Second, for the $\MYXT-\UI$ definition, the proof follows from the interchange of variables $M$ and $X$. 
\end{proof}

% \begin{definition}[Independent Sums Property]
% 	A unique information measure, \(\UI\), is said to have the indepedent sums property if when \((M_1, X_1, Y_1)\indep (M_2, X_2, Y_2)\), we have
% 	\begin{align*}
% 		&\UI(M_1, M_2 : X_1, X_2 \backslash Y_1, Y_2) \\&= \UI (M_1 : X_1 \backslash Y_1) + \UI (M_2 : X_2 \backslash Y_2 )
% 	\end{align*}
% \end{definition}

% \begin{proposition}
% 	\(\TMXY\)-UI has the independent sums property. 
% \end{proposition}

% \begin{proposition}
% 	\(\MYXT\)-UI has the independent sums property.
% \end{proposition}
% \begin{proof}
% 	The proof can be sketched as follows. We introduce a relaxation to the definition, show that the relaxed definition has the Independent Sums property, and use this to obtain the result for the original definition as well. To show the result for the relaxed definition, we demonstrate a necessary and sufficient condition on the derivative of the objective function. Proof details are available in Appendix A.2, with many gaps. 
% \end{proof}

\subsection{Counterexample to Symmetry}
Turning to symmetry however, we show that it is not in general the case that the unsymmetrized redundancy terms $R_X$ and $R_Y$ are identical. 

When the cross-covariance matrices whose kernels are involved the Gaussian Unique information expressions are full-rank, they drive the extractable unique information to zero. 

\begin{proof}[Counterexample] Consider, for the Gaussian $\TMXY-\UI$ definition, random variables $M, X, Y$ distributed as
\begin{align*}
    X, Y, M \sim \calN 
    \bracks*{0, 
        \begin{bmatrix}
            1           &0           &\rho_X \\
            0           &1           &\rho_Y\\
            \rho_X     &\rho_Y    &I
        \end{bmatrix}
    }
\end{align*}
for any $\rho_X^2 + \rho_Y^2 < 1$. This yields $\UI_X = \UI_Y = 0$, and in turn $R_X = I(M; X) \neq I(M; Y) = R_Y$. 

Likewise, for the $\MYXT-\UI$ definition, the random variables $M, X, Y$ distributed as
\begin{align*}
    X, Y, M \sim \calN 
    \bracks*{0, 
        \begin{bmatrix}
            1           &\rho           &\epsilon \\
            \rho           &1           &0\\
            \epsilon     &0    &I
        \end{bmatrix}
    }
\end{align*} 
yield terms $R_X = I(M; X) \neq I(M; Y) = R_Y$. 
\end{proof}

Thus, it is not in general true that redundant information is symmetric in either of these Markov-based approaches. 

    It is important to note here, that these counterexamples are for the Gaussian variant, which by restricting the search space of Markov chains to joint Gaussian variables, solves a different optimization problem than the general definition. 
As a result, these joint Gaussian distributions are only counterexamples to the general definitions (\ref{def:tmxy} and \ref{def:myxt}) under the assumption of Gaussian optimality for joint Gaussian distributions.

\subsection{Binary Examples}
For the simple binary operator examples displayed in Table 1, however, these definitions match our intuitive expectations, including symmetry of redundancy and synergy. 

These examples are derived with the help of \Cref{lem:binary}, which provides a strong independence condition for binary random variables. 

\begin{table*}\label{Tab: egtable}
    \linespread{1.1}
    \begin{center}
    \normalsize
    \begin{tabular}{r | c c c c}
        \toprule
        $\UI_X, \UI_Y, R, S$ &\text{RDN} &\text{UNQ} &\text{XOR} &\text{AND} \\
        \midrule
         Williams and Beer \cite{williams2010nonnegative}  &$0, 0, 1, 0$  &$1, 1, 0, 0$    &$0, 0, 0, 1$ &$0, 0, 0.311, 0.500$\\
         Bertschinger \cite{bertschinger2014quantifying}  &$0, 0, 1, 0$  &$1, 1, 0, 0$    &$0, 0, 0, 1$ &$0, 0, 0.311, 0.500$\\
         $TMXY$  &$0, 0, 1, 0$  &$1, 1, 0, 0$    &$0, 0, 0, 1$ &$0, 0, 0.311, 0.500$\\
         $MYXT$  &$0, 0, 1, 0$  &$1, 1, 0, 0$    &$0, 0, 0, 1$ &$0, 0, 0.311, 0.500$\\
        \bottomrule
    \end{tabular}
    \vspace{0.1in}
    \caption{Unique, redundant and synergistic information for four of the examples considered in~\cite{niu2019measure} and ~\cite{bertschinger2014quantifying}. In each case, the quantification of our definitions match with those in~\cite{williams2010nonnegative} and~\cite{bertschinger2014quantifying}.}
    \end{center}
\end{table*}

\section{Discussion and conclusions}

In this work, we provided two measures of unique information through optimization under Markov-relations constraints, motivated by application scenarios in fairness. A key benefit of our optimization, in contrast with commonly used measures of PID quantities, is that we obtain a random variable jointly distributed with $M,X,Y$ that represents unique information. This can be useful in understanding what function of the random variables captures unique information, which can help obtain inferences in AI and neuroscience, two of the areas where PID framework's applications are envisioned~\cite{bertschinger2014quantifying, dutta2020information, pica2017quantifying}.  We provide solutions for these optimization problems in the Gaussian case, and examine these measures using simple examples. The solutions in the Gaussian case are obtained under the assumption that the optimal solution is also Gaussian. The proof of this assumption is a part of our ongoing work. 

Possible relaxations of our definitions include cases where strict independence is not required in optimization (e.g., for unique information of $M$ in $X$, we could bound $I(T;Y)$ by a non-zero constant). Our analysis reveals that the Gaussian version of this problem can also be addressed (assuming again that the optimal distribution is Gaussian). Such relaxation also yields a version where the definition has Blackwell sufficiency, and we are examining whether it may also yield symmetry in redundant information definitions that our current definition lacks (in general). 

We do not yet understand whether our definition satisfies (what we call) the Independent Sums Property:

\begin{definition}[Independent Sums Property]
	A unique information measure, \(\UI\), is said to have the indepedent sums property if when \((M_1, X_1, Y_1)\indep (M_2, X_2, Y_2)\), we have
	\begin{align*}
		&\UI(M_1, M_2 : X_1, X_2 \backslash Y_1, Y_2) \\&= \UI (M_1 : X_1 \backslash Y_1) + \UI (M_2 : X_2 \backslash Y_2 ).
	\end{align*}
\end{definition}
We conjecture that our definitions do satisfy this property.
\begin{conjecture}
	\(\TMXY\)-UI has the independent sums property. 
\end{conjecture}

\begin{conjecture}
	\(\MYXT\)-UI has the independent sums property.
\end{conjecture}

\section{Acknowledgements}

P.~Venkatesh was supported by the Shanahan Family Foundation Fellowship at the Interface of Data and Neuroscience at the Allen Institute and the University of Washington, supported in part by the Allen Institute. We wish to thank the Allen Institute founder, Paul G.~Allen, for his vision, encouragement, and support.

\nocite{*}
\bibliographystyle{IEEEtran}
\bibliography{IEEEabrv,ref}

\appendix

\section{Appendices}
\subsection{Lemmas for the Gaussian Results}

\begin{lemma}\label{gaussian-markov}
	If \(X, W, Y\) are jointly multivariate Gaussian random variables,
	\begin{align*}
		X - W - Y \qquad \iff \qquad \Sigma_{XY} = \Sigma_{XW}\Sigma_{WW}^{-1}\Sigma_{WY}.
	\end{align*}
\end{lemma}
\begin{proof}
    We simply expand covariance terms, using the fact that the covariance matrices of conditional distributions, for jointly Gaussian variables are given by Schur complements~\cite{zhang2006schur}. 
	\begin{align*}
	    0 &= \text{Cov}(X; Y|W) = \Sigma_{XY|W}\\
	    &= \Sigma_{(XW)(YW)} / \Sigma_{WW} \\
	    &= \begin{bmatrix}
	        \Sigma_{XY} &\Sigma_{XW}\\
	        \Sigma_{WY} &\Sigma_{WW}
	    \end{bmatrix}/\Sigma_{WW}\\
	    &= \Sigma_{XY} - \Sigma_{XW}\Sigma_{WW}^{-1}\Sigma_{WY}.
	\end{align*}
	Therefore, $X\indep Y | W$ iff $\Sigma_{XY} = \Sigma_{XW}\Sigma_{WW}^{-1}\Sigma_{WY}$. 
\end{proof}

\subsection{Lemmas for Binary Examples}
\begin{lemma}\label{lem:binary}
	If \(X-Y-Z\) are finite random variables jointly distributed as \(p(x, y, z)\), and further \(Y, Z\) are binary, then if \(X \indep Z\) and \(Y \not\indep Z\), then \(X \indep Y\).
\end{lemma}
\begin{proof}
	Assume w.l.o.g. that $X, Y, Z$ are random variables over the alphabets $\{x_i\}_{i\in[l]}, \{y_i\}_{i\in[2]}, \{z_i\}_{i\in[2]}$ respectively. For every $i \in [l]$, we have 
	\begin{align}
	    p(x_i | z_1) &= \sum_{j \in[2]} p(x_i | y_j) p(y_j | z_1)\\
	    p(x_i | z_2) &= \sum_{j \in[2]} p(x_i | y_j) p(y_j | z_2).
	\end{align}
	Subtracting the second equation from the first, we obtain
	\begin{align*}
	    0 &= p(x_i | y_1) (p(y_1 | z_1) - p(y_2 | z_2))\\
	    &\quad + p(x_i | y_2) (p(y_2|z_1) - p(y_2 | z_2)). 
	\end{align*}
	Or in matrix form, 
	\begin{align*}
	    0 &= \begin{bmatrix}
	        p(x_1 | y_1) &p(x_1 | y_2)\\
	        \vdots &\vdots \\
	        p(x_l | y_1) & p(x_l | y_2) 
	    \end{bmatrix}
	    \begin{bmatrix}
	        p(y_1 | z_1) - p(y_2 | z_2)\\
	        p(y_2|z_1) - p(y_2 | z_2)
	    \end{bmatrix}.
	\end{align*}
	Now, since $Y\not \indep Z$, we know the vector $[p(y_1 | z_1) - p(y_2 | z_2), p(y_2|z_1) - p(y_2 | z_2)]^T$ is non-zero. Thus, $p[x|y]$ has nullity $\geq 1$, while being an $l\times 2$ matrix. Thus it is rank 1, and for each $i\in[l]$, we must have $p(x_i|y_1) = p(x_i|y_2)$. 
\end{proof}
\end{document}